\documentclass{article}
\usepackage{amsmath}
\usepackage{amsthm}
\usepackage{amsfonts}
\usepackage{graphicx}
\usepackage{caption}
\usepackage{enumitem}
\usepackage{hyperref}
\hypersetup{
	colorlinks=true,
	linkcolor=blue,
	filecolor=magenta,
	urlcolor=cyan,
	}
\usepackage{float}
\newtheorem{theorem}{Theorem}[section]

\newtheorem{requirements}[theorem]{Model Requirements}
\newtheorem{definition}[theorem]{Definition}
\newtheorem{proposition}[theorem]{Proposition}

\author{Will Hicks}
\begin{document}
\title{A Nonlocal Approach to The Quantum Kolmogorov Backward Equation and Links to Noncommutative Geometry}
\maketitle
\begin{abstract}
The Accardi-Boukas quantum Black-Scholes equation (\cite{AccBoukas}) can be used as an alternative to the classical approach to finance, and has been found to have a number of useful benefits. The quantum Kolmogorov backward equations, and associated quantum Fokker-Planck equations, that arise from this general framework, are derived using the Hudson-Parthasarathy quantum stochastic calculus (\cite{HP}). In this paper we show how these equations can be derived using a nonlocal approach to quantum mechanics. We show how nonlocal diffusions, and quantum stochastic processes can be linked, and discuss how moment matching can be used for deriving solutions.
\end{abstract}
\section{Introduction}
Stochastic calculus is used to model random processes for many applications (for example, see \cite{Oksendal}). Kolmogorov backward equations, and the Fokker-Planck equation, arise in the study of stochastic processes, as the partial differential equations whose solutions represent the expectation values for functions of the underlying random variable, and the probability density for the process respectively. For example, in the study of Mathematical Finance, Kolmogorov backward equations can be solved to find the risk neutral price of a derivative security, and the Black-Scholes equation is one specific example of such an equation (eg see \cite{Bjork} for more detail).
\newline
\newline
The majority of practioners in the finance industry, use models based on the application of Brownian motion, and Ito calculus. However, the application of quantum formalism to Mathematical Finance has been investigated by a number of sources. For example, see \cite{Haven_1}, \cite{Haven_2}, \cite{McCloud_1},\cite{McCloud_2}, and \cite{Segal}. Further, Accardi \& Boukas apply quantum stochastic calculus to the problem of derivative pricing, and the simulation of the financial markets (see \cite{AccBoukas}, \cite{HP}), and show how this leads to a `Quantum Black-Scholes' equation. Further analysis \& development of this method is presented in \cite{Hicks}, and \cite{Hicks2}. In particular, in \cite{Hicks2}, the author builds on techniques applied in \cite{Baaquie}, \cite{Baaquie_2}, and \cite{Linetsky}, by deriving kernel functions from quantum Kolmogorov backward equations, based on the path integral approach to quantum mechanics. The analysis shows that this can be achieved using a Hamiltonian function that is no longer a quadratic function of the momentum variable. Unfortunately, this fact leads to a number of complications.
\newline
\newline
In section \ref{intro}, we start by giving some background on quantum stochastic processes, and the quantum Black-Scholes equation. For more detail, readers can refer to \cite{AccBoukas}, and \cite{Hicks2}. Then in section \ref{GT}, we show using the example of the link between gauge transformations and changes of measure, how the non-quadratic Hamiltonian function, whilst useful in many circumstances, can lead to difficulty.
\newline
\newline
With this in mind, in section \ref{NCG} we outline the basis for a nonlocal approach. We proceed by borrowing some of the ideas \& techniques from noncommutative geometry. Noncommutative geometry was originally developed by Connes to extend the methods of algebraic geometry to the noncommutative setting (see \cite{Connes}). Sinha \& Goswami have also investigated links between quantum stochastic calculus, and noncommutative geometry in \cite{GoSi}. Noncommutative geometry is generally based on algebras of bounded operators, and results are often applied to compact manifolds. In the real world, one is generally concerned with unbounded operators on noncompact manifolds (eg $\mathbb{R}$). In this case, there are considerable difficulties in even showing that unbounded operators (which may not even share a common domain) form an algebra. However, we attempt to show in this article, that useful results can be obtained proceeding on a formal basis, using techniques inspired by noncommutative geometry. In addition to resolving complications associated with nonquadratic Hamiltonian functions, the nonlocal approach provides an alternative interpretation of solutions to the quantum Kolmogorov backward equation, or the quantum Black-Scholes equation of \cite{AccBoukas}.
\newline
\newline
In Mathematical Finance, the key financial interpretation behind the nonlocal approach rests with imposing a fundamental limit on the precision with which one can forecast a traded price in advance. For example, even if the financial market exists in a Dirac state: $\delta(x_0-x)$, we still do not know with certainty that we can fulfil a trading order at exactly this price. In section \ref{QKBE} we show how, by following this logic, the standard second order Fokker-Planck equation is transformed into the quantum Fokker-Planck equation discussed in \cite{Hicks}, and \cite{Hicks2}. We go on, in section \ref{MomMatch}, to show how the moments of a Gaussian kernel function are impacted by the introduction of the nonlocality.
\newline
\newline
The link between nonlocal diffusions and quantum stochastic processes was first discussed in \cite{Hicks}, where it is shown that given a quantum stochastic process, one can write the solution as a nonlocal diffusion. In \cite{Hicks2}, it is shown that, given a nonlocal diffusion, defined by a ``blurring'' function or ``nonlocalility'' function with defined moments, one can tailor a Riemannian metric such that the solution can be written as a quantum stochastic process. This article builds on this by showing how one can derive the quantum Fokker-Planck or quantum Kolmogorov backward equation directly from the nonlocal approach to quantum mechanics. The restriction to those nonlocality functions with defined moments, can be explained using physical principals.
\newline
\newline
In addition to providing theoretical insights into the understanding of quantum stochastic processes, it is hoped studying problems of quantum stochastic calculus using the nonlocal approach, and methods from noncommutative geometry, will provide new avenues for developing practical tools. For example, future development of calibration methods based on the associated heat kernel expansions.
\section{Quantum Stochastic Processes and the Quantum Fokker-Planck Equation}\label{intro}
In this section, we illustrate how quantum stochastic calculus can be applied to the simulation of the financial market, and derive the quantum Kolmogorov backward equation, and associated quantum Fokker-Planck equation. Then in subsequent sections, we illustrate how the same equations can be derived through using a nonlocal approach, and suggest how this can be embedded in the framework of noncommutative geometry.
\newline
\newline
The market that we are trying to model can be described by a state function sitting in the tensor product of the initial space $\mathcal{H}$ and the Boson Fock space: $\Gamma(L^2(\mathbb{R}^+,\mathcal{H}))$. This is described in more detail below.
\subsection{Initial Space}
The initial space: $\mathcal{H}$, is a Hilbert space that carries the price information from the current market. If we want to know the current price of the FTSE index, then this is represented by the operator $X$, where $X$ acts on the state function $\phi(x)$ by pointwise multiplication: $X\phi=x\phi(x)$. To get the expected price one can trade the FTSE index at right now, we carry out the following calculation:
\begin{equation}
\mathbb{E}\big[X\big]=\langle\phi,X\phi\rangle=\int_\mathbb{R} x\rvert\phi(x)\lvert^2dx, \text{for}~\phi(x)\in \mathcal{H}.
\end{equation}
If I know with certainty, that the FTSE is at $7000$, then the initial state function would be a Dirac state: $\rvert \phi(x)\lvert^2=\delta(7000-x)$.
\subsection{Boson Fock Space}
To define a quantum stochastic process we require a mechanism to incrementally amend the initial quantum state as time progresses, essentially by adding the drift and the random diffusion. This~is achieved using the Boson Fock space.
\newline
\newline
We start with functions from the time axis, with values in the Hilbert space that carries the pricing information ($\mathcal{H}$). This space is written: $\mathcal{K}=L^2(\mathbb{R}^+,\mathcal{H})$.
\newline
\newline
Next we take the exponential vectors, $\psi(f)$. For $f\in\mathcal{K}$ we have: $\psi(f)=(1,f,\frac{f\otimes f}{\sqrt{2}},....,\frac{f^{\otimes n}}{n^{1/2}},...)$, so that: $\langle \psi(f),\psi(g)\rangle=e^{\langle f,g\rangle}$, for $f,g\in \mathcal{K}$. The Boson Fock space is defined as the Hilbert space completion of these exponential vectors, which now provide the mechanism we require.
\newline
\newline
Our market state space is the tensor product space: $\mathcal{H}\otimes\Gamma(\mathcal{K})$. Initially at $t=0$, the Boson Fock space can be thought of as being empty (although this turns out to be unimportant). The operator $X$ that returns the expected FTSE price becomes $X\otimes\mathbb{I}$, where $\mathbb{I}$ represents the identity operator on the empty Boson Fock space, and the calculation of the FTSE index price right now, is unchanged.
\subsection{Quantum Drift}
A particle, with initial wave function $\phi(x)=\int_{\mathbb{R}} \tilde{\phi}(p)e^{ipx}dp$, in a system controlled by the Hamiltonian function $H(x,p)$, where $p$ represents the momentum, has a unitary time development operator: $U_t=e^{iHt}$. Thus if the operator $X_0$ returned the position at time $0$, then we have at time $t$ (we~assume Planck's constant $\hbar=1$):
\begin{equation}\label{Xt}
X_t=j_t(X_0)=U_t^*X_0U_t
\end{equation}
The Hamiltonian function $H(x,p)$ is the infinitesimal generator for the time development operator, and we can write the following quantum stochastic differential equation (SDE): 
\begin{equation}\label{drift}
dU_t=(-iHdt)U_t
\end{equation}
The situation for modelling the FTSE is exactly the same. To define a quantum SDE with drift, we~require a self-adjoint operator $H$, which controls the drift through Equations (\ref{Xt}) and (\ref{drift}).
\newline
\newline
For a classical particle with drift, the position is a deterministic function of time. Now the position of the particle is no longer deterministic. It is the wave function that evolves in a deterministic fashion.
\subsection{Quantum Diffusion}
We now add operators that allow the market state function to evolve stochastically. This is described by Hudson and Parthasarathy in~\cite{HP}. The operators we require act on the exponential vectors in the Boson Fock space as follows:
\begin{equation}
A_t\psi(g)=\bigg(\int_0^tg(s)ds\bigg)\psi(g), A^{\dagger}_t\psi(g)=\frac{d}{d\epsilon}\bigg\rvert_{\epsilon=0}\psi(g+\epsilon\chi_{(0,t)}), \Lambda_t\psi(g)=\frac{d}{d\epsilon}\bigg\rvert_{\epsilon=0}\psi(e^{\epsilon\chi_{(0,t)}}g)
\end{equation}
Further we can define the stochastic differentials as:
\begin{equation}
dA_t=\big(A_{t+dt}-A_t\big), dA^{\dagger}_t=\big(A^{\dagger}_{t+dt}-A^{\dagger}_{t}\big), d\Lambda_t=\big(\Lambda_{t+dt}-\Lambda_t\big)
\end{equation}
The significance of these operators derives from the functional form for the time development operator. In order for $U_t$ to be unitary, it must have the following form (see \cite{HP} Section 7):
\begin{equation}\label{dUt}
dU_t=-\Bigg(\bigg(iH+\frac{1}{2}L^*L\bigg)dt+L^*SdA_t-LdA^{\dagger}_t+\bigg(1-S\bigg)d\Lambda_t\Bigg)U_t
\end{equation}
where $H, L$ and $S$ are bounded linear operators on $\mathcal{H}$, with $H$ self-adjoint, and $S$ unitary. With $L=0$, and $S=1$, this reduces to the drift quantum SDE given in Equation (\ref{drift}).
\subsection{Quantum Ito Formula}
The quantum stochastic differentials can be combined using the following multiplication table (see \cite{AccBoukas} Lemma 1, and \cite{HP} Theorem 4.5):
\begin{equation}
\begin{tabular}{p{1cm}|p{1cm}p{1cm}p{1cm}p{1cm}}
-&$dA^{\dagger}_t$&$d\Lambda_t$&$dA_t$&$dt$\\
\hline
$dA^{\dagger}_t$&0&0&0&0\\
$d\Lambda_t$&$dA^{\dagger}_t$&$d\Lambda_t$&0&0\\
$dA_t$&$dt$&$dA_t$&0&0\\
$dt$&0&0&0&0
\end{tabular}
\end{equation}
We can see from the table above that:
\begin{equation}
\mathbb{E}\big[\big(\int_0^t dA_s+dA^{\dagger}_s)^2\big]=\mathbb{E}\big[\int_0^t dA_sdA_s+dA^{\dagger}_sdA^{\dagger}_s+dA_sdA^{\dagger}_s+dA^{\dagger}_sdA_s)\big]=\int_0^t ds=t=\mathbb{E}\big[W(t)^2\big].
\end{equation}
In fact, with $S=1$ in Equation (\ref{dUt}), the terms in $d\Lambda_t$ disappear. The resulting operator is commutative, and the resulting PDE is the same as the classical Black-Scholes PDE (\cite{AccBoukas} Proposition 2). For $S\neq 1$, we have a non-commutative system, and the Black-Scholes equations have more complicated dynamics. The key result, regarding the time development of $X^k$, can be obtained by application of the above multiplication rules, and is given by~\cite{AccBoukas} Lemma 1:
\begin{equation}\label{X^k}
\begin{split}
dj_t(X^k)=j_t(\lambda^{k-1}\alpha^{\dagger})dA^{\dagger}_t+j_t(\alpha\lambda^{k-1})dA_t+j_t(\lambda^k)d\Lambda_t+j_t(\alpha\lambda^{k-2}\alpha^{\dagger})dt\\
\alpha=[L^*,X]S, \alpha^{\dagger}=S^*[X,L], \lambda=S^*XS-X
\end{split}
\end{equation}
\subsection{Quantum Kolmogorov Backward Equation \& the Quantum Fokker-Planck Equation}
First, expanding a function: $u(t,j_t(X))$, as a power series, we get:
\begin{equation}
u(t,x)=\sum_{n,k\geq 0} \frac{\partial^{n+k}u}{\partial t^n\partial x^k}\Big\rvert_{t=t_0,x=x_0}(t-t_0)^n(x-x_0)^k
\end{equation}
To calculate the expected value of $u(t,X_t)$ we can apply the Quantum version of the Ito lemma from above, and collect together the terms in $dt$.
\newline
\newline
If we assume $S$, from equation (\ref{dUt}), represents a Lebesgue invariant translation, $T_{\varepsilon}$, we have(for $f(x)\in L^2(\mathbb{R})$):
\begin{equation}
\lambda f(x)=T_{-\varepsilon}XT_{\varepsilon}f(x)-Xf(x)=T_{-\varepsilon}xf(x-\varepsilon)-xf(x)=\varepsilon f(x)
\end{equation}
So we have in this case $\lambda=\varepsilon$, and the drift free Brownian motion:
\begin{equation}
dX_t=\sigma dA_t+\sigma dA_t^{\dagger}
\end{equation}
becomes instead:
\begin{equation}
dX_t=\sigma dA_t+\sigma dA_t^{\dagger}+\varepsilon d\Lambda_t
\end{equation}
Applying (\ref{X^k}) to the expansion for $du$:
\begin{equation}
du(t,X_t)=\frac{\partial u}{\partial t}dt+\sum_{k\geq 1} \frac{\partial^k u}{\partial x^k}dX_t^k
\end{equation}
and collecting together terms in $dt$, we get the quantum Kolmogorov backward equation by taking expectations. In this case:
\begin{equation}
\frac{\partial u}{\partial t}+\sum_{k\geq 2}\frac{\sigma^2\varepsilon^{k-2}}{k!}\frac{\partial^k u}{\partial x^k}=0
\end{equation}
We can derive the equivalent quantum Fokker-Planck equation, as shown in \cite{Hicks} proposition 3.1, by successive integration by parts:
\begin{equation}\label{QFP_eqn}
\frac{\partial p}{\partial t}=\sum_{k\geq 2}\frac{\sigma^2(-\varepsilon)^{k-2}}{k!}\frac{\partial^k p}{\partial x^k}
\end{equation}
\section{Gauge Transformations and Change of Measure:}\label{GT}
Here we extend the analysis, originally carried out in \cite{PHL}, to the Accardi-Boukas Quantum Black-Scholes world (see \cite{AccBoukas}, \cite{Hicks}, \cite{Hicks2}). In \cite{PHL} chapter4, Henry-Labord{\`e}re develops the classical approach to finance using the Heat Kernel on a Riemannian manifold. In section 4.5, he shows how a gauge transformation can easily be shown to be equivalent to a change of measure.
\newline
\newline
In this section we show how to apply this to the Quantum Black-Scholes equation, and how the non-quadratic Hamiltonian functions that arise from general quantum stochastic processes, lead to complications. In addition to providing insights into the physics of these Hamiltonians, it provides incentive for investigating whether quantum stochastic processes can be modelled using standard quadratic Hamiltonians. This is the subject of sections \ref{NCG} \& \ref{QKBE}.
\newline
\newline
Start from the following quantum Kolmogorov backward equation.
\begin{equation}\label{eq1.1}
\partial_{\tau}u(\tau,x)+\sigma^2\sum_{k\geq 2} \frac{\varepsilon^{(k-2)}}{k!}\partial_x^k u(\tau,x)=0
\end{equation}
This equation can be derived from the following Hamiltonian (see \cite{Hicks2}):
\begin{equation}\label{eq1.2}
\hat{H} = \sigma^2\sum_{k\geq 2} \frac{\varepsilon^{k-2}\hat{P}^k}{k!}=\frac{\sigma^2}{\varepsilon^2}\Big(exp(\varepsilon\hat{P})-\varepsilon\hat{P}-1\Big)
\end{equation}
By using the transformation: $(\tau,x)=(it,-iy)$, and the relation: $\hat{P}=i\partial_y=\partial_x$, we get:
\begin{equation}\label{eq1.3}
\hat{H} = \sigma^2\sum_{k\geq 2} \frac{\varepsilon^{k-2}}{k!}\partial_x^k
\end{equation}
We now apply a gauge transformation: $u'(\tau,x)=e^{\Lambda(x)}u(\tau,x)$. We get:
\begin{equation}\label{eq1.4}
\partial_x u'(\tau,x)=e^{\Lambda(x)}\partial_x u(\tau,x)+\partial_x\Lambda(x) u(\tau,x)
\end{equation}
If we set $\hat{P}'=i\partial_y-i\partial_y\Lambda(iy)$, then we get:
\begin{equation}\label{eq1.5}
\hat{P}'u'=\hat{P}u
\end{equation}
\begin{equation}\label{eq1.6}
\hat{P}'=\hat{P}+v(x)
\end{equation}
So, at a classical level, we find that applying the gauge transformation is equivalent to adding a function of $x$ to the momentum variable:$p'=p+v(x)$. The classical Hamiltonian is given by:
\begin{equation}
H(p)=\frac{\sigma^2}{\varepsilon^2}\Big(exp(\varepsilon p)-\varepsilon p-1\Big)
\end{equation}
Inserting: $p=p'-v(x)$, we get:
\begin{equation}
H(p')=H(p',x)=\frac{\sigma^2}{\varepsilon^2}\Big(exp(-\varepsilon v(x))exp(\varepsilon p')-\varepsilon p'+\varepsilon v(x)-1\Big)
\end{equation}
We now carry out the Legendre transformation in order to derive a formula for the velocity, in terms of the momentum. Differentiating: $p'\dot{x}-H(p',x)$, with respect to $p'$, and setting to zero, we get the classical velocity in the new coordinate system:
\begin{equation}
\dot{x}'=\frac{\sigma^2}{\varepsilon}\Big(exp(-\varepsilon v(x))exp(\varepsilon p')-1\Big)
\end{equation}
So, the canonical momentum is given by:
\begin{equation}
p'_0(\dot{x},x)=\frac{\sigma^2}{\varepsilon}ln\Big(\frac{\varepsilon\dot{x}}{\sigma^2}+1\Big)+v(x)
\end{equation}
We find that, the translated canonical momentum: $p'_0(\dot{x},x)$ is given by translating the original canonical momentum: $p_0(\dot{x})$ as before. However, note the following:
\begin{itemize}
\item The velocity: $\dot{x}$, has a nonlinear relationship with $p_0$, and so the gauge transformation does not lead to a simple translation of the velocity.
\item The Lagrangian is no longer a quadratic function of $\dot{x}$.
\end{itemize}
We will see, that in the Quantum case, $\varepsilon\neq 0$, this leads to difficulty in terms of showing that the gauge transformation leads to a simple drift change.
\subsection{Classical Case, $\varepsilon=0$}
In this case we have: $\dot{x}=\sigma^2p$ and so $\dot{x}'=\sigma^2\big(p-v(x)\big)$. The Lagrangian becomes:
\begin{equation}
L'(\dot{x}',x)=\frac{(\dot{x}'+\sigma^2v(x))^2}{2\sigma^2}
\end{equation}
The resulting integral kernel: $K^0_T(x)'$ is given by (where $\mathcal{D}x$ represents the Feynman path integral):
\begin{equation}
K^0_T(x)'=\int_{\infty}^{\infty} exp\bigg(-\int_0^T \frac{\big(\dot{x}'+\sigma^2v(x)\big)^2}{2\sigma^2}dt\bigg)\mathcal{D}x
\end{equation}
So we still end up with a Gaussian kernel function, but with a drift that depends on $x$. Importantly, $L'(\dot{x}',x)=L(\dot{x}+\sigma^2v(x))$. Where $\Lambda(x)=cx$, we end up with a constant drift. This mirrors the analysis presented in \cite{PHL} chapter 4.5. The financial implications are that a change of measure, or a change of risk free numeraire, can be achieved using the gauge transformation.
\subsection{Quantum Case, $\varepsilon\neq 0$}
Now, we have:
\begin{equation}
\dot{x}'=\dot{x}-\sigma^2v(x)+\sum_{k\geq 2}\frac{\varepsilon^{k-1}\sigma^2v(x)^k}{k!}
\end{equation}
Therefore, the new Lagrangian becomes:
\begin{equation}
L'(\dot{x}',x)=\sum_{k\geq 0} \frac{(-\varepsilon)^k\Big(\dot{x}'+\sigma^2v(x)-\sum_{k\geq 2}\frac{\varepsilon^{k-1}\sigma^2v(x)^k}{k!}\Big)^{(k+2)}}{\sigma^{2(k+1)}(k+1)(k+2)}
\end{equation}
Therefore, unlike the classical case, $L'(\dot{x}',x)\neq L(\dot{x}+\sigma^2v(x))$. It is no longer immediately clear that the integral kernel function post gauge transformation: $K^{\varepsilon}_T(x)'$ is a simple translation of $K^{\varepsilon}_T(x)$.
\section{Outline of the Nonlocal Approach}\label{NCG}
\subsection{Spectral Triples}
The framework of noncommutative geometry enables one to write the physical laws arising from non-Abelian gauge field theories, using the language and tools of differential geometry. See for example \cite{VanSuij}. Noncommutative geometry is based on the spectral triple, which is defined by:
\begin{itemize}
\item A Hilbert space $\mathcal{H}$.
\item A unital $C^*$ algebra: $\mathcal{A}$, represented as bounded operators on $\mathcal{H}$.
\item A self-adjoint operator $\mathcal{D}$, such that the resolvent $(I+\mathcal{D})^{-1}$ is a compact operator and $[\mathcal{D},a]$ is bounded for $a\in\mathcal{A}$.
\end{itemize}
For example, the basic example of a spectral triple (see \cite{VanSuij}, chapter 4.3) is the canonical triple associated with a compact Riemannian spin manifold, $M$:
\begin{itemize}
\item $\mathcal{A}=C^{\infty}(M)$, the algebra of smooth functions on $M$.
\item $\mathcal{H}=L^2(S)$ of square integrable sections of spinor bundle $S\rightarrow M$.
\item $\mathcal{D_M}$ the Dirac operator with associated connection.
\end{itemize}
Once a suitable spectral triple has been defined, the key mathematical tool for understanding the behaviour of a system is the heat kernel expansion (see \cite{VanSuij}, chapter 7.2):
\begin{equation}
Tr\Big(e^{-tD^2}\Big)=\sum_{\alpha} t^{\alpha}c_{\alpha}
\end{equation}
Readers interested in the application of noncommutative geometry to quantum stochastic calculus should refer to \cite{GoSi}, where the authors outline a number of examples of heat semi-groups on noncommutative spaces, and noncommutative spectral triples.
\newline
\newline
In the next section, proceeding in a mainly formal basis, we suggest ways in which this framework can be used to develop the integral kernels used for the modelling the probability spaces defined by quantum stochastic differential equations. This in turn provides an alternative way to view quantum diffusion processes as nonlocal diffusions.
\subsection{A Nonlocal Approach}\label{NL}
One of the key steps in derivations of the Feynman Path Integral (for example see \cite{Folland}, chapter 8), involves the definition of so called ``generalised states'': $\vert x_0\rangle=\delta(x_0-x)$, to represent a particle situated at position $x_0$ with probability $1$. In reality however, it is rarely possible to know this information with complete precision. For example, even if I see a price quote for a traded asset on a Bloomberg screen, it is still reasonably unlikely that I can fulfil my full order by trading at exactly this price. Usually, for $\vert x_0\rangle$, we can measure the expected price as:
\begin{equation}
\mathbb{E}[X]=\langle x_0\vert Xx_0\rangle=\int_M y\delta(x_0-y)dy=x_0
\end{equation}
We can incorporate the uncertainty around what price we will actually achieve when executing the trade, by changing the operator to incorporate a probability distribution, $H(y)$:
\newline
\newline
The operator $X$ becomes: $(X,H)$, and rather than $X\psi(x)=x\psi(x)$, we have:
\begin{equation}
(X,H)\psi(x)=x\int_M\psi(x-y)H(y)dy
\end{equation}
To calculate the expected price we can trade at, given the market is state: $\psi(x)$, we get:
\begin{equation}
\mathbb{E}[X]=\int_M\overline{\psi(y)}\int_{\mathbb{R}}y\psi(y-u)H(u)dudy
\end{equation}
In effect, even if we know that the current market price is precisely $x_0$, there is still some uncertainty over the actual price I trade at. This uncertainty is introduced into the model using the function $H(y)$.
\newline
\newline
In terms of the spectral triple, one possible approach would be for the canonical triple over the manifold $M$, with 1 dimensional spinor bundle: $(C^{\infty}(M),L^2(M),\partial_x+\mathcal{A}_x)$ to become: $(C^{\infty}(M),L^2(M,H(y)dy),\partial_x+\mathcal{A}_x)$, where $\mathcal{A}_x$ represents the connection.
\newline
\newline
However, the $C^*$ algebra: $C^{\infty}(M)$ no longer forms an algebra over the Hilbert space: $L^2(M,H(y)dy)$. To rectify this, we define the spectral triple:
\begin{equation}
(C^{\infty}(M)\otimes \Gamma(E),\mathcal{H},(\partial_x,H))
\end{equation}
Where:
\begin{itemize}
\item $\Gamma(E)$ represents a fibre bundle, whereby each section is a continuous function with values in the space of probability distributions over the manifold $M$: $C_0(M;S)$ for a space of probability distributions: $S$.
\item $\mathcal{H}$ represents a dense subset of $L^2(M)$.
\item $H\in S$.
\item $(f(x),H(x))\in C^{\infty}(M)\otimes S$ acts on $\psi(x)$ by: $(f,H)\psi(x)=f(x)\int_M \psi(x-y)H(y)dy$.
\item The Dirac operator: $(\partial_x,H)$ acts on $\psi(x)$ by: $\partial_x\Big(\int_M \psi(x-y)H(y)dy\Big)$.
\end{itemize}
We explain this in more detail in \ref{NST}.
\subsection{A Noncommutative Spectral Triple}\label{NST}
\subsubsection{Action of $C^{\infty}(\mathbb{R})$ by pointwise multiplication}
The first point to consider, is that spectral triples are defined in relation to $C^*$ algebras of bounded operators, whereas many of the operators we require are unbounded. Alternatively, spectral triples are often defined on a {\it compact} manifold M, whereas the majority of real life examples (for example traded financial underlyings) exist on noncompact manifolds, such as the real number line: $\mathbb{R}$.
\newline
\newline
Whilst noting this critical point, the objective in this article is to investigate how the general framework of noncommutative geometry could apply. Therefore, for now we sweep this issue aside and carry on the investigation using the ``algebra'' of smooth functions: $C^{\infty}(\mathbb{R})$, with the aim of generating results that can be used for real world applications.
\newline
\newline
Next, if we view the ``spectral triple'' as: $(C^{\infty}(\mathbb{R}),L^2(\mathbb{R},H(y)dy),\partial_x+\mathcal{A}_x)$, with $C^{\infty}(\mathbb{R})$ acting on the Hilbert space by pointwise multiplication, then we find that:
\begin{equation}
\Big(a(x)b(x)\Big)\circ \psi(x)=a(x)b(x)\int_{\mathbb{R}}\psi(x-y)H(y)dy
\end{equation}
However:
\begin{equation}
a(x)\circ \Big(b(x)\psi(x)\Big)=a(x)\int_{\mathbb{R}}b(x-y_1)\psi(x-y_1-y_2)H(y_1)H(y_2)dy_1dy_2\neq \Big(a(x)b(x)\Big)\circ \psi(x)
\end{equation}
So we find that $C^{\infty}(\mathbb{R})$ cannot form an algebra under composition of operators. Since this is a crucial ingredient of the path integral construction, we do not use this approach.
\subsubsection{Action of $C^{\infty}(\mathbb{R})\otimes \Gamma(E)$ on a dense subset of $L^2(\mathbb{R})$}
First, we denote the sections of $\Gamma(E)$ as: $H(z;x)$. This function, returns a distribution: $H\in S$ for each position: $x$ on the underlying manifold (in this case the real numbers, $\mathbb{R}$). Define the action of $(a(x),H(z;x))$ on a dense subset: $\mathcal{H}$ of $L^2(\mathbb{R})$ by:
\begin{equation}\label{rep}
(a(x),H(z;x))\psi(x)=a(x)\int_{\mathbb{R}}\psi(x-y)H(y;x)dy
\end{equation}
Now, we consider the bilinear map, which we define as follows:
\begin{definition}\label{bl_map}
\begin{equation}
\begin{split}
\Big(a(x),H_a(z;x)\Big)\Big(b(x),H_b(z;x)\Big)=\Big(a(x),H_{ab}(z;x)\Big)\\
H_{ab}(z;x)=\int_{\mathbb{R}} H_a(u;x)b(x-u)H_b(z-u;x-u)du
\end{split}
\end{equation}
\end{definition}\ \\
Using definition \ref{bl_map}, we find that the composition of different operators now works:
\begin{equation}
\bigg(\Big(a(x),H_a(z;x)\Big)\Big(b(x),H_b(z;x)\Big)\bigg)\circ \psi(x)=a(x)\int_{\mathbb{R}}\int_{\mathbb{R}} b(x-u)\psi(x-z)H_a(u;x)H_b(z-u;x-u)dudz
\end{equation}
Set: $z=y_1+y_2$, and we get:
\begin{equation}
\begin{split}
=a(x)\int_{\mathbb{R}}\int_{\mathbb{R}} b(x-y_1)\psi(x-y_1-y_2)H_a(y_1;x)H_b(y_2;x-y_1)dy_1dy_2\\
=(a(x),H_a(z;x))\circ\Big(b(x)\int_{\mathbb{R}}\psi(x-y)H_b(y;x)dy\Big)
\end{split}
\end{equation}
We also note that, for $C^{\infty}(\mathbb{R})\otimes \Gamma(E)$ to form a unital algebra, the space $S$ must include the Dirac delta, since under (\ref{bl_map}) we have:
\begin{equation}
\begin{split}
(1,\delta)(a,H)=(1,\int_{\mathbb{R}}\delta(u)a(x-u)H(z-u;x-u)du)\\
=(1,a(x)H(z;x))=(a,H)\\
(a,H)(1,\delta)=(a,\int_{\mathbb{R}}\delta(z-u)H(u;x)du)=(a,H)
\end{split}
\end{equation}
In order to show that: $C^{\infty}(\mathbb{R})\otimes \Gamma(E)$ forms a noncommutative $C^*$ algebra we still need to show:
\begin{enumerate}
\item[1)] $A(BC)\circ\psi=(AB)C\circ \psi$ under the bilinear map (\ref{bl_map}), and the action (\ref{rep}).
\item[2)] We can define a norm $\vert\vert..\vert\vert$, such that our algebra is complete in the norm topology.
\item[3)] We can define an involution $*$ such that $\vert\vert A^*A\vert\vert=\vert\vert A\vert\vert^2$
\end{enumerate}
\underline{Proof of 1):}
\newline
\newline
Using the bilinear map \ref{bl_map}, we have:
\begin{equation}
\Big(b(x),H_b(z;x)\Big)\Big(c(x),H_c(z;x)\Big)=\Big(b(x),\int_{\mathbb{R}} c(x-y)H_b(y;x)H_c(z-y;x-y)dy
\end{equation}
So inserting this into $A(BC)$, where $A=(\big(a(x),H_a(z;x)\Big)$, and so on for $B, C$, we get:
\begin{equation}\label{A(BC)}
\begin{split}
\Big(a(x),H_a(z;x)\Big)\circ\bigg(\Big(b(x),H_b(z;x)\Big)\Big(c(x),H_c(z;x)\Big)\bigg)\\
=\Big(a(x),\int_{\mathbb{R}}\int_{\mathbb{R}}b(x-u)c(x-y_1-y_2)H_a(u;x)H_b(y_2;x-y_1)H_c(z-y_1-y_2;x-y_1-y_2)dy_2dy_1\Big)
\end{split}
\end{equation}
Similarly, we have:
\begin{equation}
\Big(a(x),H_a(z;x)\Big)\Big(b(x),H_b(z;x)\Big)=\Big(a(x),\int_{\mathbb{R}} b(x-u)H_a(u;x)H_b(z-u;x-u)du\Big)
\end{equation}
So inserting this into: $(AB)C$, we get:
\begin{equation}\label{(AB)C}
\begin{split}
\bigg(\Big(a(x),H_a(z;x)\Big)\Big(b(x),H_b(z;x)\Big)\bigg)\circ\Big(c(x),H_c(z;x)\Big)\\
=\Big(a(x),\int_{\mathbb{R}}b(x-y_1)c(x-y_2)H_a(y_1;x)H_b(y_2-y_1;x-y_1)H_c(z-y_2;x-y_2dy_1dy_2\Big)
\end{split}
\end{equation}
Exchanging the order of integration, and replacing $y_2'=y_2+y_1$ we see that equations (\ref{A(BC)}) and (\ref{(AB)C}) match.
\newline
\newline
\underline{Proof of 2):}
\newline
\newline
We can use the standard operator norm. If $A=(a,H)$, we have:
\begin{definition}
Where: $\vert\vert\psi\vert\vert$ is given by $\sqrt{\langle\psi\vert\psi\rangle}$, we have:
\begin{equation}
\vert\vert A\vert\vert=\sup_{\vert\vert\psi\vert\vert=1} \vert\vert A\psi\vert\vert
\end{equation}
\end{definition}\ \\
In our case:
\begin{equation}
\vert\vert A\vert\vert=\sup_{\vert\vert\psi\vert\vert=1} \int_{\mathbb{R}} \overline{A\psi(y)}A\psi H(y)dy
\end{equation}
Assuming we restrict $H(Y)$ to a dense subset of $L^2(\mathbb{R})$, as described in the proof of 1) above, completeness follows from the underlying commutative $C^*$ algebra.
\newline
\newline
\underline{Proof of 3):}
\newline
\newline
For the commutative $C^*$ algebra of complex valued smooth functions on a Riemannian manifold: $C^{\infty}(M)$, we have that $a(x)^*$ is given by the complex conjugate: $\overline{a(x)}$. This can be carried over into the noncommutative $C^*$ algebra. So $(a,H)^*$ becomes $(\overline{a(x)},H)$.
\subsubsection{Defining $\mathcal{H}$, and $S$}
We have the following basic requirements, in order for our model to be useful in practice:
\begin{requirements}\label{model_req}\
\begin{enumerate}
\item[i)] For all $H\in S$, and all $\psi\in \mathcal{H}$, the convolution: $H\ast\psi$ exists. If this condition is not met, then we need to restrict the distributions in $S$, in order that we can carry out necessary calculations.
\item[ii)] The Hilbert space $\mathcal{H}$ is dense in $L^2(\mathbb{R})$. If this condition is not met, then there will be valid market states, that we cannot represent in our Hilbert space.
\end{enumerate}
\end{requirements}\ \\
It turns out that this can be achieved, by restricting the space of distributions: $S$ to those probability distributions, where the moments are defined, and this in turn links these noncommutative ``spectral triples'' to the quantum stochastic processes discussed by Hudson \& Parthasarathy in \cite{HP}.
\begin{proposition}\label{prop_1}
For the ``spectral triple'': $\big(C^{\infty}(\mathbb{R})\otimes \Gamma(E),\mathcal{H},(\partial_x,H)\big)$, if $\mathcal{H}$ is dense in $L^2(\mathbb{R})$ then the space of distributions meeting requirements \ref{model_req} consists only of those distributions: $H(x)$ such that the moments:
\begin{equation}
\mu_i(H)=\int_{\mathbb{R}}y^iH(y)dy
\end{equation}
exist for all $i\geq 0$.
\end{proposition}
\begin{proof}
Consider the dense subspace $\mathcal{K}\subset L^2(\mathbb{R})$ consisting of those functions, that can be expanded as a power series: $\psi(x)=\sum_{n\geq 0} a_nx^n$. Requirement \ref{model_req} i) means that the integral:
\begin{equation}\label{I1}
\sum_{n\geq 0}a_n\int_{\mathbb{R}} (x-y)^nH(y;x)dy=\sum_{n\geq 0}a_n{n\choose i}\sum_{i\leq n}(-1)^ix^{n-i}\int_{\mathbb{R}} y^iH(y;x)dy
\end{equation}
must exist. Now assume further that $N$ is the lowest integer such that $\mu_N(H)$ doesn't exist. The existence of (\ref{I1}) implies that there must exist a polynomial, $P(x)$ of degree $N$, such that:
\begin{equation}
\int_{\mathbb{R}} P(y)H(y;x)dy=S_N<0
\end{equation}
exists. We can write $P(x)=\sum_{i\leq N}b_ix^i$ for complex $b_i$. Now, we write $p_{N-1}(x)=\sum_{i\leq N-1}b_ix^i$. Since all moments less than $N$ exist we have:
\begin{equation}
\int_{\mathbb{R}} P_{N-1}(y)iH(y;x)dy=S_{N-1}<0
\end{equation}
Now we have: $\mu_N(H)=\frac{S_N-S_{N-1}}{b_N}<\infty$, which is a contradiction. Therefore, if $\mathcal{H}=\mathcal{K}$, those functions in $L^2(\mathbb{R})$ that can be expanded as a power-series, then $S$ must consists only of distributions such that all the moments exist.
\newline
Therefore, assume we choose $\mathcal{H}\neq \mathcal{K}$. Now we assume we have: $f\in\mathcal{H}$ such that:
\begin{equation}
\int_{\mathbb{R}}f(y)H(y;x)dy<\infty
\end{equation}
even though the moments of $H(z;x)$ do not exist. Since $\mathcal{K}$ is dense in $L^2(\mathbb{R})$, and $\mathcal{H}\subset L^2(\mathbb{R})$, there must exist $g(x)\in \mathcal{K}$, such that: $\vert\vert g(x) - f(x)\vert\vert<\epsilon$, for arbitrarily small $\epsilon$. By the triangle inequality we have:
\begin{equation}
\bigg\vert\int_{\mathbb{R}}f(y)H(y;x)dy-\int_{\mathbb{R}}g(y)H(y;x)dy\bigg\vert\leq\int_{\mathbb{R}}\vert f(y)-g(y)\vert H(y;x)dy<\epsilon
\end{equation}
However, we have that if all moments for $H(z;x)$ are not defined, then:
\begin{equation}
\int_{\mathbb{R}}g(y)H(y;x)dy
\end{equation}
is not defined. This is a contradiction, and the proposition is proved.
\end{proof}
\section{Nonlocal Diffusions \& the Quantum Kolmogorov Backward Equation}\label{QKBE}
\subsection{A Nonlocal Derivation of the Quantum Kolmogorov Backward Equation}
In this section, and in section \ref{PathInt}, we follow standard steps, for example following those described in \cite{Hall}, in deriving a nonlocal formulation of the quantum Fokker-Planck equation and associated integral kernel function.
\newline
\newline
As noted above, there are a number of issues in extending rigorous results on spectral triples, and noncommutative geometry to unbounded operators on noncompact manifolds. We therefore proceed on a purely formal basis with the aim of generating useful results. We start with the usual Hamiltonian function for a free particle of mass $m$:
\begin{equation}\label{Hamiltonian}
\hat{H}(\hat{p})=\frac{\hat{p}^2}{2m}
\end{equation}
Now, rather than the usual definition, $\hat{p}\psi=i\partial_x$ we set $\hat{p}\psi=(i\partial_x,h)$, and we have:
\begin{equation}
\begin{split}
\hat{p}\psi(x)=i\partial_x\bigg(\int_{\mathbb{R}}h(x-y)\psi(y)dy\bigg)\\
=i\partial_x\bigg(\int_{\mathbb{R}}h(y)\psi(x-y)dy\bigg)\\
=h\ast i\partial_x\psi(x)
\end{split}
\end{equation}
Therefore, inserting this into (\ref{Hamiltonian}) we have:
\begin{equation}\label{Hamitonian2}
\begin{split}
\hat{H}\psi=h\ast h\ast \frac{-1}{2m}\frac{\partial^2\psi}{\partial x^2}\\
=\frac{-1}{2m}\frac{\partial^2}{\partial x^2}\bigg(\int_{\mathbb{R}}H(x-y)\psi(y)dy\bigg)\\
H=h\ast h
\end{split}
\end{equation}
Inserting this into the Schr{\"o}dinger equation, we get:
\begin{equation}
i\frac{\partial\psi}{\partial t}=\frac{1}{2m}\frac{\partial^2}{\partial x^2}\bigg(\int_{\mathbb{R}}H(x-y)\psi(y)dy\bigg)
\end{equation}
After carrying out the usual Wick rotation: $t=i\tau$ we get our nonlocal Fokker-Planck equation:
\begin{equation}\label{FP_NL}
\frac{\partial\psi}{\partial\tau}+\frac{1}{2m}\frac{\partial^2}{\partial x^2}\bigg(\int_{\mathbb{R}}H(x-y)\psi(y)dy\bigg)=0
\end{equation}
We know from proposition \ref{prop_1}, that the moments of $h$ must exist. This in turn implies that the moments of $h\ast h$ must also exist. Therefore, we can expand (\ref{FP_NL}) using a Kramers-Moyal expansion:
\begin{equation}\label{FP_KM}
\begin{split}
\frac{\partial\psi}{\partial\tau}+\frac{1}{2m}\frac{\partial^2}{\partial x^2}\bigg(\sum_{k\geq 0} \frac{(-1)^k\mu_H^k}{k!}\frac{\partial^{(k+2)}\psi}{\partial x^{(k+2)}}\bigg)=0\\
\mu_H^k=\int_{\mathbb{R}} y^kH(y)dy
\end{split}
\end{equation}
Following, \cite{Hicks2}, we can write $s=\int_{x_0}^{x} \frac{dy}{\sqrt{g(y)}}$ to get:
\begin{equation}
\mu_H^k=\int_{\mathbb{R}} y^kH(y)\frac{dy}{\sqrt{g(y)}}
\end{equation}
Finally, with infinite degrees of freedom in the metric function: $g(y)$, we can use a moment matching algorithm to solve:
\begin{equation}
\mu_H^k=\frac{2\varepsilon^k}{(k+1)(k+2)}
\end{equation}
Inserting this into equation (\ref{FP_KM}), we get back to the quantum Fokker-Planck equation: (\ref{QFP_eqn}). Thus, on a formal level, combining proposition 1 from \cite{Hicks2}, we have shown that given a nonlocal diffusion, we can find a quantum stochastic process, and that these 2 processes are intrinsically linked.
\subsection{The Fundamental Solution}\label{PathInt}
In this section, we follow a standard method in deriving the fundamental solution to the Sch{\"o}dinger equation, that can be used as an integral kernel function (for example see \cite{Hall}, \cite{Hicks2}), but based on the the nonlocal Hamiltonian function: (\ref{Hamitonian2}). For a given Hamiltonian: $\hat{H}$, the Schr{\"o}dinger equation is:
\begin{equation}\label{Schrodinger}
i\frac{\partial\psi}{\partial t}=\hat{H}\psi
\end{equation}
Following, \cite{Hall} chapter 4, we start by assuming $\psi(x,t)$ has the form:
\begin{equation}\label{test_fn}
\psi(x,t)=exp(i(px-\omega(p)t)
\end{equation}
Using Hamiltonian: (\ref{Hamitonian2}), and inserting (\ref{test_fn}) we get:
\begin{equation}
\begin{split}
\omega (p)\psi=-\frac{1}{2m}\frac{\partial^2}{\partial x^2}\bigg(\int_{\mathbb{R}}H(y)exp\Big(i\big(p(x-y)-\omega (p)t\Big)dy\bigg)\\
=-\frac{1}{2m}\frac{\partial^2}{\partial x^2}\bigg(e^{i\big(px-\omega (p)t\big)}\int_{\mathbb{R}}H(y)e^{-ipy}dy\bigg)\\
=-\frac{1}{2m}\frac{\partial^2}{\partial x^2}\bigg(e^{i\big(px-\omega (p)t\big)}\widetilde{H}(p)\bigg)
\end{split}
\end{equation}
So finally, we end up with:
\begin{equation}
\omega (p)\psi=\frac{p^2}{2m}\widetilde{H}(p)\psi
\end{equation}
We can now use this function to calculate the required non-Gaussian kernel function.
\begin{proposition}\label{Kernel}
Let the Hamiltonian be given by: (\ref{Hamitonian2}). Then the fundamental solution to the Schr{\"o}dinger equation is given by:
\begin{equation}
K^{H}_t(x,t)=\frac{1}{\sqrt{2\pi}}\mathcal{F}^{-1}\bigg(exp\Big(\frac{-ip^2\widetilde{H}(p)t}{2m}\Big)\bigg)
\end{equation}
For initial conditions: $\psi_0(x)\in L^2(\mathbb{R})\cap L^1(\mathbb{R})$, the solution to the Schr{\"o}dinger equation is given by:
\begin{equation}
\psi(x,t)=\frac{1}{\sqrt{2\pi}}\psi_0\ast K^{H}_t
\end{equation}
\end{proposition}
\begin{proof}
The proof follows the same steps outlined in \cite{Hall}, Theorem 4.5. First note, that: $\int_{\mathbb{R}}K^{H}_tdp$, solves the Sch{\"o}dinger equation.
\newline
\newline
Also, we have that $\mathcal{F}(K^{H}_t1_{[-n,n]}\ast\psi_0)=\sqrt{2\pi}\mathcal{F}(K^{H}_t1_{[-n,n]})\mathcal{F}\psi_0$.
\newline
\newline
$\mathcal{F}(K^{H}_t1_{[-n,n]})$ is bounded, and converges pointwise to: $\frac{1}{\sqrt{2\pi}}exp\Big(\frac{-ip^2\widetilde{H}(p)t}{2m}\Big)$. This is enough to guarantee that $K^{H}_t1_{[-n,n]}\ast\psi_0$ must converge in $L^2(\mathbb{R})$.
\newline
\newline
We have shown above that $K^{H}_t$ solves the Schr{\"o}dinger equation, and it is easy to see that: $K^{H}_0\psi_0=\psi_0$. Therefore $\psi(x,t)=\frac{1}{\sqrt{2\pi}}\psi_0\ast K^{H}_t$ is the solution required.
\end{proof}
Given initial conditions: $\psi_0(x)$, the solution to equation (\ref{Schrodinger}) can be written: $exp(-it\hat{H})\psi_0(x)$. If we switch now to imaginary time, equation (\ref{Schrodinger}) becomes the Kolomogorov backward equation. Writing $u(x,\tau)$ rather than $\psi(x,\tau)$, $\sigma^2=\frac{1}{m}$, and $\tau=it$, we have:
\begin{equation}
\frac{\partial u}{\partial \tau}+\frac{\sigma^2}{2}\frac{\partial^2}{\partial x^2}\bigg(\int_{\mathbb{R}} H(x-y)u(y,\tau)dy\bigg)=0
\end{equation}
We get for a small time step $\delta\tau$, $u(x,\delta\tau)=exp(-\delta\tau\hat{H})u(x,0)$, and so using proposition \ref{Kernel}: 
\begin{equation}
u(x,\delta\tau)=\frac{1}{2\pi}\int_{\mathbb{R}} K^H_{\delta\tau}(x-x_0)u(x_0,0)dx_0
\end{equation}
\subsection{A Moment Matching Algorithm}\label{MomMatch}
In practice, for a particular choice of the volatility parameter: $\sigma$, and the nonlocality function $H(x)$, there are 3 existing potential methods for the calculation of solutions to the quantum Fokker-Planck equation, or the associated quantum Kolmogorov Backward equation:
\begin{enumerate}
\item[i)] One can use the particle Monte-Carlo method, as described in \cite{Hicks} section 4.
\item[ii)] One can calculate the value for an integral kernel using the results described in \cite{Hicks2}, proposition 4.
\item[iii)] One could use a numerical calculation of the inverse Fourier transform to evaluate the kernel function from \ref{Kernel}.
\end{enumerate}
Unfortunately, each of these methods have drawbacks. Method i) is generally even slower than conventional Monte-Carlo methods, owing to the additional steps required to calculate the impact of the nonlocality function $H(x)$. Furthermore, the method must retain memory of the ongoing position of each Monte-Carlo path throughout the simulation. This is in contrast to conventional Monte-Carlo simulations, where each path can be simulated in isolation to the other paths.
\newline
\newline
The principal drawback with using the Kernel functions defined in \cite{Hicks2} proposition 4, and in proposition \ref{Kernel} above, relate to the instability of the integrals involved, and the resulting difficulty in their numerical approximation. Therefore, in this section we outline a moment matching algorithm, that provides the possibility for fast \& robust approximation of the kernel functions. This method is based on the following proposition.
\begin{proposition}\label{moments}
Let the moments for the nonlocalility function $H(x)$, be given by: $a_n$. Then the moments for the Kernel function described in proposition \ref{Kernel} are given by:
\begin{equation}
\begin{split}
\mu_1=0\\
n\geq 2, \mu_n=\sum_j\frac{n!}{2(\#P^j_n)!}\prod_{i\in P^j_n} (\sigma^2\tau)a_{i-2}
\end{split}
\end{equation}
Where: $P^j_n$ represent the partitions of $n$ without using the number $1$, and $\#P^j_n$ represents the number of elements in the partition.
\end{proposition}
\begin{proof}
We have that $\widetilde{H}(p)$ represents the characteristic function for the distribution $H(x)$:
\begin{equation}
\widetilde{H}(p)=\int_{\mathbb{R}}e^{ipx}H(x)dx
\end{equation}
Therefore, given the moments: $a_n$, we can write:
\begin{equation}\label{char_H}
\widetilde{H}(p)=\sum_{j\geq 0}\frac{a_j(ip)^j}{j!}
\end{equation}
Furthermore, from proposition \ref{Kernel}, the Fourier transform of the kernel we are after is given (up to normalising constant) by:
\begin{equation}\label{char_K}
\mathcal{F}\Big(K^{H}_t(x,t)\Big)=exp\Big(\frac{-ip^2t\widetilde{H}(p)}{2m}\Big)
\end{equation}
Inserting (\ref{char_H}) into (\ref{char_K}) we get:
\begin{equation}
\mathcal{F}\Big(K^{H}_t(x,t)\Big)=exp\Big(\frac{-ip^2t}{2m}\sum_{j\geq 0}\frac{a_j(ip)^j}{j!}\Big)
\end{equation}
Performing the Wick rotation $\tau=it$, and using the notation $\sigma^2=1/m$, gives:
\begin{equation}
\mathcal{F}\Big(K^{H}_{\tau}(x,\tau)\Big)=exp\Big(-\frac{\sigma^2p^2\tau}{2}\sum_{j\geq 0}\frac{a_j(ip)^j}{j!}\Big)
\end{equation}
Therefore the moment generating function for $K^{H}_{\tau}$ is given by:
\begin{equation}
M_{K^{H}_{\tau}}(p)=exp\Big(\frac{p^2\sigma^2\tau}{2}\sum_{j\geq 0}\frac{a_j(p)^j}{j!}\Big)
\end{equation}
Expanding out the powers of $p$ gives:
\begin{equation}
M_{K^{H}_{\tau}}(p)=1+\sum_{k\geq 2}\bigg(\sum_j\frac{1}{2(\#P^j_k)!}\prod_{i\in P^j_n} (\sigma^2\tau)a_{i-2}\bigg)p^k
\end{equation}
The result then follows from the definition of the moment generating function.
\end{proof}
Once the moments of a probability distribution are known, there are numerous numerical methods that can be applied to finding the final probability density function. For example see \cite{Mnat}, and \cite{TekelCohen}. We defer further investigation of these methods to a future study.
\subsubsection{Illustrative Example}
In this subsection, we briefly illustrate some results where $H(x)$ is a normal distribution with variance given by $\varepsilon^2$. In this case we have $a_{2n-1}=0, a_{2n}=\epsilon^{2n}(2n-1)!!$, where $(2n-1)!!=(2n-1)(2n-3)...1$.
\newline
\newline
Plugging this into proposition \ref{moments}, we get:
\begin{itemize}
\item Since by assumption, the partitions do not include $1$, $P^j_2$ consists only of the set: $\{2\}$, so $\mu_2=\sigma^2\tau$
\item $\mu_3=0$
\item $P^j_4$ consists of: $\{4\}$, and $\{2,2\}$ so $\mu_4=3(\sigma^2\tau)^2+12(\sigma^2\tau)\varepsilon^2$
\end{itemize}
So finally we see that applying a Gaussian nonlocality function $H(x)$ with zero mean, and variance $\epsilon^2<<\sigma^2\tau$, increases the Kurtosis of the resulting kernel function by an amount: $12(\sigma^2\tau)\epsilon^2$. The standard deviation is not impacted. In effect, $H(x)$ has given the kernel function ``fat tails''.
\newline
\newline
We also note the crucial role of the ratio $(\epsilon^2/\sigma^2\tau)$. The smaller this ratio, the smaller the impact from the nonlocality on the kurtosis of the kernel function. As $\tau\rightarrow 0$, we find that the increase in kurtosis becomes more and more pronounced. $\mu_4$ decreases with $O(\tau)$ rather than $O(\tau^2)$ as would be the case for a standard Gaussian kernel function. This mirrors the numerical results shown in \cite{Hicks2}, section 5, where it was found that the non-Gaussian kernel functions tended to the standard Gaussian for longer time to maturities.
\section{Conclusion}
In this article we have shown how to derive the quantum Kolmogorov backward equation through a nonlocal formulation of quantum mechanics. This builds on results obtained in \cite{Hicks}, and \cite{Hicks2} to show that there are deep links between nonlocal diffusions, and quantum stochastic processes.
\newline
\newline
We have suggested how to amend the canonical spectral triple that is used in the formulation of conventional quantum mechanics using the framework of noncommutative geometry. Although there are significant obstacles in extending the framework to unbounded operators, the approach leads to useful ways of interpreting equations, and potential new avenues for developing analytic and numerical methods for real world applications.
\bibliographystyle{amsplain}

\begin{thebibliography}{99}

\bibitem{AccBoukas} Luigi Accardi, Andreas Boukas:
The Quantum Black-Scholes Equation, {\em Global Journal of Pure and Applied Mathematics}, (2006) vol.2, no.2, pp. 155-170.
\bibitem{Baaquie} B.E. Baaquie:
Quantum mechanics, path integrals and option pricing: reducing the complexity of finance, {\em Nonlinear physics: theory and experiment, II} (Gallipoli, 2002), 33-339, World Sci. Publ., River Edge, NJ, 2003. MR2028802
\bibitem{Baaquie_2}
Baaquie, B.E. Statistical microeconomics and commodity prices: Theory and empirical results. {\em Phil. Trans. R. Soc. A} \textbf{2016}, {\em 374}, 20150104,
doi:10.1098/rsta.2015.0104.
\bibitem{Bjork} Tomas Bjork:
Arbitrage Theory in Continuous Times. Oxford University Press, Third Edition, 2009
\bibitem{Connes} Alain Connes:
Noncommutative Geometry. Boston, MA: Academic Press, ISBN 978-0-12-185860-5
\bibitem{Folland} Gerald B. Folland:
Quantum Field Theory, A Tourist Guide for Mathematicians. American Mathematical Society, Mathematical Survey and Monographs, Volume 149
\bibitem{GoSi} Kalyan B. Sinha, Debashish Goswami:
Quantum Stochastic Processes and Noncommutative Geometry. Cambridge University Press, Cambridge Tracts in Mathematics, 2007.
\bibitem{Hall} Brian C. Hall:
Quantum Theory for Mathematicians. Springer Graduate Texts in Mathematics 267
\bibitem{Haven_1}
Haven, E. A Discussion on Embedding the Black-Scholes Option Pricing Model in a Quantum Physics Setting. {\em Physica A} {\bf 2002}, {\em 304}, 507--524.
\bibitem{Haven_2}
Haven, E. A Black-Scholes Schr{\"o}dinger Option Price: Bit versus qubit. {\em Physica A} \textbf{2003}, \emph{324}, 201--206.
\bibitem{PHL} Pierre Henry-Labord{\`e}re:
Analysis, Geometry and Modelling in Finance, Advanced Methods in Option Pricing. Chapman \& Hall/CRC Financial Mathematics Series
\bibitem{Hicks}
Hicks, W. Nonlocal Diffusions and the Quantum Black-Scholes Equation: Modelling the Market Fear Factor.  {\em Commun. Stoch. Anal. } {\bf 2018}, {\it  12}, 109--127.
\bibitem{Hicks2}
Hicks, W. PT Symmetry, Non-Gaussian Path Integrals, and the Quantum Black-Scholes Equation. {\em Entropy}, {\bf 2019}, 21(2), 105.
\bibitem{HP} Hudson, R.L.; Parthasarathy, K.R:
Quantum Ito's Formula and Stochastic Evolutions. {\em Commun Math. Phys.} \textbf{1984}, \emph{93}, 301--323.
\bibitem{Linetsky} Vadim Linetsky:
The Path Integral Approach to Financial Modelling and Options Pricing, {\em Computational Economics}, 11: 129-163, 1998.
\bibitem{McCloud_1} Paul McCloud:
In Search of Schr{\"o}dinger's Cap, available at $https://papers.ssrn.com/sol3/papers.cfm?abstract_id=2341301$
\bibitem{McCloud_2} Paul McCloud:
Quantum Bounds for Option Prices, available at $https://arxiv.org/abs/1712.01385$
\bibitem{Mnat}
Robert M. Mnatsakanov and Artak S. Hakobyan. Recovery of Distributions via Moments, {\em IMS Lecture Notes-Monograph Series,. Optimality: The Third Erich L. Lehmann Symposium}, Vol. 57 (2009) 252-265.
\bibitem{Oksendal} Bernt Oksendal:
Stochastic Differential Equations, An Introduction with Applications. Fifth Edition, Springer-Verlag Jan 1998
\bibitem{Segal} Segal, W.; Segal, I.E.
The Black-Scholes pricing formula in the quantum context. {\em Proc. Natl. Acad. Sci. USA} \textbf{1998}, \emph{95}, 4072--4075.
\bibitem{TekelCohen}
J.Tekel and L. Cohen. Constructing and estimating probability distributions from moments, Proc. SPIE 8391, 83910E, 2012.
\bibitem{VanSuij} Walter D. van Suijlekom:
Noncommutative Geometry and Particle Physics. Springer Mathematical Physics Studies, 2015

\end{thebibliography}


\end{document}